\documentclass[a4paper]{article}

\usepackage[english]{babel}
\usepackage[utf8x]{inputenc}
\usepackage[T1]{fontenc}

\usepackage[a4paper,top=3cm,bottom=2cm,left=3cm,right=3cm,marginparwidth=2cm]{geometry}

\usepackage{graphicx}
\usepackage{url}
\usepackage[colorinlistoftodos]{todonotes}
\usepackage[colorlinks=true, allcolors=blue]{hyperref}
\usepackage{adjustbox}
\usepackage{nicefrac}
\usepackage{amsmath,amsfonts,amssymb,amsthm}
\usepackage{subcaption}

\newtheorem{theorem}{\bf Theorem}[section]
\newtheorem{lemma}{\bf Lemma}[section]

\newtheorem{corollary}{\bf Corollary}[section]
\newtheorem{example}{\bf Example}[section]

\graphicspath{ {./images/} }

\title{\textbf{Measures of contextuality in cyclic systems and the negative probabilities measure \(\text{CNT}_{3}\)}}
\author{Giulio Camillo$^1$\footnote{giulio.silva@usp.br}\ , Víctor H. Cervantes$^2$}
\date{%
    {\textit{$^1$Instituto de Física, Universidade de São Paulo, Brasil}}\\%
    {\textit{$^2$Deparment of Psychology, University of Illinois Urbana-Champaign, USA}}\\[2ex]%
    \today
}

\begin{document}
%
%
%

\maketitle

\begin{abstract}
Several principled measures of contextuality have been proposed for general systems of random variables (i.e. inconsistentlly connected systems). 
The first of such measures was based on quasi-couplings using negative probabilities (here denoted by \(\text{CNT}_{3}\), Dzhafarov \& Kujala, 2016).
Dzhafarov and Kujala (2019) introduced a measure of contextuality, \(\text{CNT}_{2}\), 
that naturally generalizes to a measure of non-contextuality.
Dzhafarov and Kujala (2019) additionally conjectured that in the class of cyclic systems these two measures are proportional.
Here we prove that that conjecture is correct.
Recently, Cervantes (2023) showed the proportionality of \(\text{CNT}_{2}\)
and the Contextual Fraction measure (CNTF) introduced by Abramsky, Barbosa, and Mansfeld (2017). 
The present proof completes the description of the interrelations of all contextuality measures as they pertain to cyclic systems.
\end{abstract}



Contextuality is a property of systems of random
variables. 
A system is contextual when the observed joint distributions
within different contexts are incompatible with 
the equality in probability of variables across contexts
(we shall give a formal definition of contextuality below).
In the contextuality literature, several measures or indexes
of the degree of contextuality of a system have been introduced.
Each of these measures reflects a unique aspect of contextuality
and together provide a pattern of the system's contextuality.
The class of cyclic systems is prominent in applications of contextuality
and is used to represent many important scenarios of
quantum contextuality,
such as the EPR/Bohm scenario \cite{Fine.1982.Joint,Clauser.1969.Proposed},
or the Klyachko-Can-Binicioğlu-Shumovsky scenario \cite{Klyachko.2008.Simple}.

For cyclic systems, it has been conjectured that all measures in the literature are
proportional to each other.
In Ref.~\cite{Kujala.2019.Measures}, the equality of three of these measures (\(\text{CNT}_{1}, \text{CNT}_{2}\), and \(\text{CNT}_{0}\))
was proved,
and in Ref.~\cite{Cervantes.2023.note}, the proportionality
of \(\text{CNT}_{2}\) and the Contextual Fraction (\(\text{CNTF}\))
was demonstrated.
Together, these two results show the proportionality of all but one
of the measures found in the literature.
The measure \(\text{CNT}_{3}\) based on negative probabilities
was conjectured to be proportional to \(\text{CNT}_{2}\) in Ref.~\cite{Dzhafarov.2019.Contextuality-by-Default}.
In this paper, we prove the truth of this conjecture by means
of showing that \((n - 1)\text{CNT}_{3} = \text{CNTF}\).
Thus, this paper culminates the theoretical description of
the interrelations of all contextuality measures as they pertain cyclic systems.

\section{Contextuality-by-Default}

In this section, we present the Contextuality-by-Default
(CbD) approach to contextuality analysis
\cite{Dzhafarov.2016.Contextuality-by-Default,Dzhafarov.2017.Contextuality-by-Default,Dzhafarov.2016.Context-content,Kujala.2019.Measures,Dzhafarov.2019.Contextuality-by-Default,Dzhafarov.2020.Contextuality}.
A \emph{system} of random variables is a set of double-indexed random
variables \(R_{q}^{c}\), where \(c \in C\) is the \emph{context} of the random variable, 
the conditions under which it is recorded,
and \(q \in Q\) denotes its \emph{content},
the property of which the random variable is a measurement.
The following is a presentation of a system:
\begin{equation}
\mathcal{R} = \left\{ R_{q}^{c} : c \in C, q \in Q, q \prec c \right\}, \label{eq:generalSystem}
\end{equation}
where \(q \prec c\) denotes that content \(q\) is measured in context \(c\).

For each \(c \in C\), the subset
\begin{equation}
\mathrm{R}^{c} = \left\{ R_{q}^{c} : q \in Q, q\prec c \right\} \label{eq:bunch}
\end{equation}
is referred to as the \emph{bunch} for context \(c\). 
The variables within a bunch are jointly distributed.
That is, bunches are random vectors with a given probability distribution. 
For each \(q\in Q\), the subset 
\begin{equation}
\mathcal{R}_{q} = \left\{ R_{q}^{c} : c\in C, q\prec c \right\} \label{eq:connection}
\end{equation}
is referred to as the \emph{connection} corresponding to content \(q\).
However, no two random variables within a connection \(\mathcal{R}_{q}\) are
jointly distributed; 
thus, they are said to be \emph{stochastically unrelated}.\footnote{More generally,
any two \(R_{q}^{c}, R_{q'}^{c'} \in \mathcal{R}\) with
\(c \neq c'\) are stochastically unrelated.
We emphasize that variables within a connection (and within a system) are
stochastically unrelated by using calligraphic script for their names,
and that variables of a bunch do possess a joint distribution
by using roman script.}

Cyclic systems are a prominent class of systems of random variables.
They are the object of Bell's theorem \cite{Bell.1964.Einstein,Bell.1966.problem},
the Leggett--Garg theorem \cite{Leggett.1985.Quantum},
Suppes and Zanotti's theorem \cite{Suppes.1981.When}, 
the Klyachko-Can-Binicioğlu-Shumovsky theorem \cite{Klyachko.2008.Simple},
as well as many other theoretical results (see e.g., \cite{Araujo.2013.All,Kujala.2016.Proof}). 
Cyclic systems are used to model most applications that empirically explore contextuality
(e.g., \cite{Adenier.2017.Test,Asano.2014.Violation,Cervantes.2018.Snow,Dzhafarov.2015.there,Hensen.2015.Loophole-free,Wang.2013.Quantum}).
Furthermore, as shown in Refs.~\cite{Dzhafarov.2020.Contextuality,Abramsky2013},
a system without cyclic subsystems is necessarily noncontextual.
A system \(\mathcal{R}\) is said to be \emph{cyclic} if
\begin{enumerate}
\item each of its contexts contains two jointly distributed \emph{binary} random variables,
\item each content is measured in two contexts, and
\item there is no proper subsystem of \(\mathcal{R}\) that satisfies (i) and (ii).
\end{enumerate}
The number \(n\) of contexts (and contents) on 
a cyclic system is known as its \emph{rank}. 
For any cyclic system, a rearrangement and numbering of its contexts and contents
can always be found so that the system
can be given the presentation
\begin{equation}
\mathcal{R}_{n} = \left\{ \left\{ R_{i}^{i}, R_{i \oplus 1}^{i} \right\} : i = 1, \ldots, n \right\} ,\label{eq:cyclicSystem}
\end{equation}
where \(R_{j}^{i}\) stands for \(R_{q_{j}}^{c_{i}}\), and \(\oplus 1\)
denotes cyclic shift \(1 \mapsto 2, \ldots, n-1 \mapsto n, n \mapsto 1\).\footnote{Similarly,
\(\ominus 1\) will denote the inverse shift of \(\oplus 1\).}
In this way, the variables \(\left\{ R_{i}^{i}, R_{i \oplus 1}^{i} \right\}\)
constitute the bunch corresponding to context \(c_{i}\). 
The following
matrices depict the format of two cyclic systems: a cyclic system of rank 3,
and a cyclic system of rank 6.
\begin{equation}\label{eq:cycles}
\begin{array}{ccc}
\begin{array}{|c|c|c||c|}
\hline R_{1}^{1} & R_{2}^{1} &           & c_{1} \\
\hline           & R_{2}^{2} & R_{3}^{2} & c_{2} \\
\hline R_{1}^{3} &           & R_{3}^{3} & c_{3} \\
\hline
\hline q_{1}     & q_{2}     & q_{3}     & \mathcal{R}_{3} \\
\hline 
\end{array} 
&  & 
\begin{array}{|c|c|c|c|c|c||c|}
\hline R_{1}^{1} & R_{2}^{1} &           &           &           &           & c_{1} \\
\hline           & R_{2}^{2} & R_{3}^{2} &           &           &           & c_{2} \\
\hline           &           & R_{3}^{3} & R_{4}^{3} &           &           & c_{3} \\
\hline           &           &           & R_{4}^{4} & R_{5}^{4} &           & c_{4} \\
\hline           &           &           &           & R_{5}^{5} & R_{6}^{5} & c_{5} \\
\hline R_{1}^{6} &           &           &           &           & R_{6}^{6} & c_{6} \\
\hline
\hline q_{1}     & q_{2}     & q_{3}     & q_{4}     & q_{5}     & q_{6}     & \mathcal{R}_{6} \\
\hline 
\end{array}
\end{array}
\end{equation}

A system is said to be \emph{consistently connected} if, for any two
\(R_{q}^{c}, R_{q}^{c'}\), the respective distributions of \(R_{q}^{c}\)
and \(R_{q}^{c'}\) coincide. 
This property of a system encodes the \emph{no-disturbance}
or the \emph{no-signaling} requirements in quantum physics,
and corresponds to the \emph{marginal selectivity} 
condition in the selective influences literature in psychology.
If this property is not satisfied, the system is said to be \emph{inconsistently connected.} 
In general, a system of random variables \(\mathcal{R}\)
is inconsistently connected.

Kujala and Dzhafarov \cite{Kujala.2019.Measures}
introduced a vectorial description of a system.
This representation is obtained by taking the probabilities
of events for each random variable, for each bunch, and for each connection
in the system. 
Let \(\mathbf{l}_{(.)}\), \(\mathbf{b}_{(.)}\), and \(\mathbf{c}_{(.)}\)
denote these vectors, respectively. For describing a cyclic system,
the first two vectors are:
\begin{equation}
\mathbf{l}_{(.)} =  \left(
\begin{array}{c}
\Pr(R_{i}^{i} = r_{i}^{i}) \\
\Pr(R_{i \oplus 1}^{i} = r_{i \oplus 1}^{i})
\end{array}
\right)_{i = 1, \ldots, n}, \label{eq:fullMargins}
\end{equation}
\begin{equation}
\mathbf{b}_{(.)} =  \left(
\begin{array}{c}
\Pr(R_{i}^{i} = r_{i}^{i}, R_{i \oplus 1}^{i} = r_{i \oplus 1}^{i})
\end{array}
\right)_{i = 1, \ldots, n}, \label{eq:fullBunches}
\end{equation}
where \(r_{i}^{i}, r_{i \oplus 1}^{i}\) may each take one of two values
(we will assume  \(r_{i}^{i}, r_{i \oplus 1}^{i} = 0, 1\)).

The vector \(\mathbf{c}_{(.)}\) contains imposed probabilities.
These probabilities define a \emph{coupling} of the variables within each connection.
A coupling of a set of random variables \(\left\{ X_{i} \right\}_{i \in I}\),
where \(I\) indexes the variables in the set, is a new set of jointly
distributed random variables \(\left\{ Y_{i} \right\}_{i\in I}\) such that for each \(i\in I\), 
the distributions of \(X_{i}\) and \(Y_{i}\) coincide. 
In a \emph{multimaximal coupling} of a set of random variables, 
for any two \(Y_{i}, Y_{i'}\), the probability \(\Pr(Y_{i} = Y_{i'})\) is 
the maximal possible given their individual distributions. 
If we denote the variables of the multimaximal coupling of connection \(\mathcal{R}_{q_{j}}\) by
\begin{equation}
\mathrm{T}_{q_{j}} = \left\{ T_{j}^{i} : c_{i} \in C, q_{j} \prec c_{i} \right\}, \label{eq:connectionCoupling}
\end{equation}
one obtains the vector
\begin{align}
\mathbf{c}_{(.)} =\left(
\Pr(T_{i}^{i} = r_{i}^{i}, T_{i}^{i \ominus 1} = r_{i}^{i \ominus 1})
\right)_{i = 1, \ldots, n}, \label{eq:fullConnection}
\end{align}
with \(r_{i}^{i}, r_{i}^{i \ominus 1} = 0, 1\), and where 
\begin{equation}
\begin{array}{ll}
\Pr(T_{i}^{i} = 0, T_{i}^{i \ominus 1} = 0) & = \min(\Pr(R_{i}^{i} = 0), \Pr(R_{i}^{i \ominus 1} = 0)), \\
\Pr(T_{i}^{i} = 1, T_{i}^{i \ominus 1} = 1) & = \min(\Pr(R_{i}^{i} = 1), \Pr(R_{i}^{i \ominus 1} = 1)), \\
\Pr(T_{i}^{i} = 0, T_{i}^{i \ominus 1} = 1) & = \Pr(R_{i}^{i \ominus 1} =1 ) - \Pr(T_{i}^{i} = 1, T_{i}^{i \ominus 1} = 1), \\
\text{and} \\
\Pr(T_{i}^{i}= 1 , T_{i}^{i \ominus 1} = 0) & =\Pr(R_{i}^{i} = 1) - \Pr(T_{i}^{i} = 1, T_{i}^{i \ominus 1} = 1).
\end{array}
\end{equation}
Clearly, in the multimaximal coupling \(\mathrm{T}_{q_{j}}\),
\(\Pr(T_{i}^{i} = 0, T_{i}^{i \ominus 1} = 1) = 0\) or
\(\Pr(T_{i}^{i} = 1, T_{i}^{i \ominus 1} = 0) = 0\).
Note that whenever a system \(\mathcal{R}\)
is consistently connected, then for any two \(R_{q}^{c}, R_{q}^{c'}\),
the corresponding variables of a multimaximal coupling of \(\mathcal{R}_{q}\)
are almost always equal (that is, \(\Pr(T_{q}^{c} = T_{q}^{c'}) = 1\)).

\subsection{Consistification}

\emph{Consistification} is a procedure that can be applied to any system of binary random
variables that will create a new system \(\mathcal{R}^{\ddagger}\) that
is consistently connected and whose contextual status is the same
as that of the original system \(\mathcal{R}\). 
Here we introduce the procedure closely following the presentation given in Ref.~\cite{Cervantes.2023.note}.
The consistification of system \(\mathcal{R}\) is obtained by constructing a new system
\(\mathcal{R}^{\ddagger}\) in the following manner. 
First, define the set of contents \(Q^{\ddagger}\) of the new system as
\begin{equation}
Q^{\ddagger} = \left\{ q_{ij} : c_{i} \in C, q_{j} \in Q, q_{j} \prec c_{i} \right\} .
\end{equation}
That is, for each content \(q_{j}\) and each of the contexts \(c_{i}\)
in which it is recorded, we define a content \(q_{ij}\)=``\(q_{j}\)
recorded in context \(c_{i}\)''. 
Next, define the new set of contexts \(C^{\ddagger}\) as
\begin{equation}
C^{\ddagger} = C \sqcup Q,
\end{equation}
the disjoint union of the contexts and the contents of the system \(\mathcal{R}\). 
Then, define the new relation 
\begin{equation}
\prec^{\ddagger} = \left\{ (q_{ij}, c_{i}) : q_{j} \in Q, c_{i} \in C, q_{j} \prec c_{i} \right\} 
\sqcup 
\left\{ (q_{ij}, q_{j}) : q_{j} \in Q, c_{i} \in C, q_{j} \prec c_{i}\right\}.
\end{equation}
That is, the new content \(q_{ij}\) is recorded in precisely two of
the new contexts, \(c_{i}, q_{j} \in C^{\ddagger}\).
Therefore, the bunch
\begin{equation}
\mathrm{R}^{c_{i}}=\left\{ R_{q_{ij}}^{c_{i}}:q_{ij}\in Q^{\ddagger},q_{ij}\prec^{\ddagger}c_{i}\right\} 
\end{equation}
coincides with the bunch
\begin{equation}
\mathrm{R}^{c_{i}} = \left\{ R_{q}^{c_{i}} : q \in Q, q \prec c_{i} \right\} 
\end{equation}
of the original system; while the bunch
\begin{equation}
\mathrm{R}^{q_{j}} = \left\{ R_{q_{ij}}^{q_{j}} : q_{ij} \in Q^{\ddagger}, q_{ij} \prec^{\ddagger} q_{j} \right\} 
\end{equation}
is constructed by defining new jointly distributed random variables
\(\left\{ R_{q_{ij}}^{q_{j}} \right\}_{q_{ij} \prec^{\ddagger} q_{j}}\)
such that \(\mathrm{R}^{q_{j}}\) is the multimaximal coupling of \(\mathcal{R}^{q_{j}}\)
of system \(\mathcal{R}\).

In particular if \(\mathcal{R}\) is a cyclic system of rank \(n\), then
its consistified system \(\mathcal{R}^{\ddagger}\) is a consistently
connected cyclic system of rank \(2n\). 
The following matrices show the consistification of system \(\mathcal{R}_{3}\) 
and how its bunches relate to the bunches of the original system and the multimaximal
couplings of its connections.
\begin{equation}
\label{eq:consistification}
\begin{array}{ccc}
\begin{array}{|c|c|c|c|c|c||c|}
\hline 
  R_{q_{11}}^{c_{1}} & R_{q_{12}}^{c_{1}} &                    &                    &                    &                    & c_{1} \\ \hline  
                     & R_{q_{12}}^{q_{2}} & R_{q_{22}}^{q_{2}} &                    &                    &                    & q_{2} \\ \hline  
                     &                    & R_{q_{22}}^{c_{2}} & R_{q_{23}}^{c_{2}} &                    &                    & c_{2} \\ \hline 
                     &                    &                    & R_{q_{23}}^{q_{3}} & R_{q_{33}}^{q_{3}} &                    & q_{3} \\ \hline
                     &                    &                    &                    & R_{q_{33}}^{c_{3}} & R_{q_{31}}^{c_{3}} & c_{3} \\ \hline 
  R_{q_{11}}^{q_{1}} &                    &                    &                    &                    & R_{q_{31}}^{q_{1}} & q_{1} \\ \hline\hline 
  q_{11}             & q_{12}             & q_{22}             & q_{23}             & q_{33}             & q_{31}             & \mathcal{R}_{3}^{\ddagger} \\
\hline 
\end{array} 
&  & 
\begin{array}{|c|c|c|c|c|c||c|}
\hline 
  R_{1}^{1} & R_{2}^{1} &           &           &           &           & c_{1} \\ \hline\hline  
            & T_{2}^{1} & T_{2}^{2} &           &           &           & q_{2} \\ \hline\hline  
            &           & R_{2}^{2} & R_{3}^{2} &           &           & c_{2} \\ \hline\hline 
            &           &           & T_{3}^{2} & T_{3}^{3} &           & q_{3} \\ \hline\hline 
            &           &           &           & R_{3}^{3} & R_{1}^{3} & c_{3} \\ \hline\hline 
  T_{1}^{1} &           &           &           &           & T_{1}^{3} & q_{1} \\ \hline 
\end{array}
\end{array}
\end{equation}

\subsection{Linear programs for contextuality and its magnitude}

When it is possible to find a coupling
of \(\mathcal{R}\) that simultaneously agrees with both the bunches and
with the multimaximal couplings of the connections,
the system is said to be \emph{noncontextual}
\cite{Dzhafarov.2016.Contextuality-by-Default,Dzhafarov.2017.Contextuality-by-Default}.
This coupling, if it exists, can be found as the solution to a system
of linear equations defined using vectors \(\mathbf{l}_{(.)}\), \(\mathbf{b}_{(.)}\),
and \(\mathbf{c}_{(.)}\). 
First, let \(\mathbf{p}_{(.)}\) be the concatenation
of the three vectors \(\mathbf{l}_{(.)}\), \(\mathbf{b}_{(.)}\), and \(\mathbf{c}_{(.)}\). 
Then, consider a vector \(\mathbf{s}\) of length \(2^{2n}\) of all possible values 
that a coupling of the entire system \(\mathcal{R}_{n}\) can take. 
That is, if we denote a coupling of \(\mathcal{R}_{n}\) by \(\mathrm{S}_{n}\),
then \(\mathrm{S}_{n}\) is a random vector whose possible values are the conjunction of events
\[
\left\{ S_{j}^{i} = r_{j}^{i} : i, j = 1, \ldots, n, \ q_{j} \prec c_{i} \right\},
\]
with \(r_{j}^{i} = 0, 1\). 
These values are the components of \(\mathbf{s}\).
Let \(\mathbf{M}_{(.)}\) be an incidence \((0/1)\) matrix with \(2^{2n}\)
columns labeled by the elements of \(\mathbf{s}\) and \(12n\) rows labeled
by the events whose probabilities are the components of \(\mathbf{p}_{(.)}\).
The cells of \(\mathbf{M}_{(.)}\) are filled as follows:
\begin{itemize}
\item If the \(u\)th component of \(\mathbf{p}_{(.)}\) is \(\Pr(R_{j}^{i} = r_{j}^{i})\)
and the \(v\)th component of \(\mathbf{s}\) includes the event 
\(\left\{ S_{j}^{i} = r_{j}^{i} \right\}\),
then the cell \((u, v)\) of \(\mathbf{M}_{(.)}\) is a \(1\);
\item if the \(u\)th component of \(\mathbf{p}_{(.)}\) is 
\(\Pr(R_{i}^{i} = r_{i}^{i}, R_{i \oplus i}^{i} = r_{i \oplus 1}^{i})\)
and the \(v\)th component of \(\mathbf{s}\) includes the event 
\(\left\{ S_{i}^{i} = r_{i}^{i}, S_{i \oplus 1}^{i} = r_{i \oplus 1}^{i} \right\}\),
then the cell \((u,v)\) is a \(1\);
\item if the \(u\)th component of \(\mathbf{p}_{(.)}\) is 
\(\Pr(T_{i}^{i} = r_{i}^{i}, T_{i}^{i \ominus 1} = r_{i}^{i \ominus 1})\)
and the \(v\)th component of \(\mathbf{s}\) includes the event
\(\left\{ S_{i}^{i} = r_{i}^{i}, S_{i}^{i \ominus 1} = r_{i}^{i \ominus 1} \right\}\),
then the cell \((u,v)\) is a \(1\);
\item all other cells have zeroes.
\end{itemize}
A detailed description of the construction of \(\mathbf{M}_{(.)}\)
for general systems of binary random variables can be found in
Ref.~\cite{Dzhafarov.2016.Context-content}.

The system \(\mathcal{R}_{n}\) described by \(\mathbf{p_{(.)}^{*}}\)
is noncontextual \cite{Dzhafarov.2016.Context-content} if and only
if there is a vector \(\mathbf{h}\geq0\) (component-wise) such that
\begin{equation}
\mathbf{M}_{(.)} \mathbf{h} = \mathbf{p_{(.)}^{*}}. \label{eq:largeLFT}
\end{equation}
Any solution \(\mathbf{h}^{*}\) gives the probability distribution
of a coupling \(\mathrm{S}_{n}\) of \(\mathcal{R}_{n}\) that contains as its
marginals both the bunch distributions and the multimaximal couplings
of the connections of \(\mathcal{R}_{n}\).

Clearly, the rows of \(\mathbf{M}_{(.)}\) are not linearly
independent.
The vectorial description can therefore be reduced by taking only a subset of
the components of \(\mathbf{p}_{(.)}\),
such that the corresponding rows of \(\mathbf{M}_{(.)}\) are linearly independent. 
Let us choose the following reductions of the vectors
\(\mathbf{l}_{(.)}\), \(\mathbf{b}_{(.)}\), and \(\mathbf{c}_{(.)}\) 
and denote them 
\(\mathbf{l}\), \(\mathbf{b}\), and \(\mathbf{c}\):
\begin{equation}
\label{eq:reducedVectors}
\begin{array}{cllc}
\mathbf{l} & = \left( p_{j}^{i} \right)_{i, j = 1, \ldots, n, \ q_{j} \prec c_{i}} 
& = \left( \Pr(R_{j}^{i} = 1) \right)_{i, j = 1, \ldots, n, \ q_{j} \prec c_{i}}, \\
\mathbf{b} & = \left( p_{i, i \oplus 1} \right)_{i = 1, \ldots, n} 
& = \left( \Pr(R_{i}^{i} = 1, R_{i \oplus 1}^{i} = 1) \right)_{i = 1 , \ldots, n}, 
&   \text{and} \\
\mathbf{c} & = \left( p^{i, i \ominus 1} \right)_{i = 1, \ldots, n} 
& = \left( \Pr(T_{i}^{i} = 1, T_{i}^{i \ominus 1} = 1) \right)_{i = 1, \ldots, n}.
\end{array}
\end{equation}

Linear programming tasks
to compute the CbD-based measures of contextuality can be defined
by employing these vectorial representations
\cite{Kujala.2019.Measures}.
Let a system \(\mathcal{R}_{n}\) be described by 
\begin{equation}
\mathbf{p}^{*} = \left(
\begin{array}{c}
\mathbf{l^{*}} \\
\mathbf{b^{*}} \\
\mathbf{c^{*}}
\end{array}
\right),
\end{equation}
where \(\mathbf{l^{*}}\) and \(\mathbf{b^{*}}\) are the empirical probabilities of the system, 
and \(\mathbf{c}^{*}\) are the probabilities found from the multimaximal couplings of each of its connections.
Let \(\mathbf{M}\) be the incidence matrix found by taking the rows of \(\mathbf{M}_{(.)}\)
corresponding to the elements of \(\mathbf{p}^{*}\). 
Note that the system is noncontextual if and only if there is a vector \(\mathbf{h}\geq0\)
(component-wise) such that 
\begin{equation}
\mathbf{Mh} = \mathbf{p}^{*}, \label{eq:LFT}
\end{equation}
subject to \(\mathbf{1}^{\intercal} \mathbf{h} = 1\) \cite{Dzhafarov.2016.Context-content}.
Denoting the rows of \(\mathbf{M}\) that correspond to \(\mathbf{l}^{*}\), \(\mathbf{c}^{*}\), \(\mathbf{b}^{*}\)
by, respectively, \(\mathbf{M_{l}}\), \(\mathbf{M_{b}}\), \(\mathbf{M_{c}}\),
we can rewrite (\ref{eq:LFT}) as
\begin{equation}
\label{eq:extensoLFT}
\left(
\begin{array}{c}
\mathbf{M_{l}} \\
\mathbf{M_{b}} \\
\mathbf{M_{c}}
\end{array}
\right)
\mathbf{h} = \left(
\begin{array}{c}
\mathbf{l^{*}} \\
\mathbf{b^{*}} \\
\mathbf{c^{*}}
\end{array}
\right).
\end{equation}
An example matrix \(\mathbf{M}_{(.)}\) for cyclic systems of rank 2
can also be found in Ref.~\cite{Dzhafarov.2016.Context-content},
whereas Ref.~\cite{Cervantes.2023.note} illustrates matrix \(\mathbf{M}\)
for cyclic systems of rank 4.

Let \(\mathbf{M'} = \left(\mathbf{M} | \mathbf{M} \right)\),
\(\mathbf{y'} = (\mathbf{y}_{+}^{\intercal} | -\mathbf{y}_{-}^{\intercal})^{\intercal}\), and
\(\mathbf{y} = \mathbf{y}_{+} - \mathbf{y}_{-}\),
where \(\mathbf{y}_{+}, \mathbf{y}_{-}\) are vectors of
\(2^{2n}\) nonnegative components.
Clearly, 
\[
    \mathbf{M'y'} =
\mathbf{My}.
\]
The contextuality measure \(\text{CNT}_{3}(\mathcal{R}_{n})\) can 
be computed solving the linear programming task \cite{Dzhafarov.2016.Context-content}:
\begin{equation}
\label{lp:CNT3}
\begin{tabular}{|ccc|}
\hline
find           & minimizing                            & subject to \tabularnewline
\hline 
\(\mathbf{y'}\) & \(\mathbf{1}^{\intercal} \mathbf{y}_{-}\) & 
\(\mathbf{M'y'}
 = \mathbf{p}^{*}\) \tabularnewline
& & \(\mathbf{1}^{\intercal} 
\mathbf{y'}
= 1\)                 \tabularnewline 
& & \(\mathbf{y}_{+}, \mathbf{y}_{-} \geq 0\)         \tabularnewline 
\hline 
\end{tabular}.
\end{equation} 
For any solution \(\mathbf{y'^{*}}\), 
we compute \(\mathbf{y}^{*} = \mathbf{y}_{+}^{*} - \mathbf{y}_{-}^{*}\) and
\(\text{CNT}_{3}(\mathcal{R}_{n}) = \left\Vert \mathbf{y}^{*}\right\Vert _{1} - 1\).
For brevity and due to the uniqueness of the Hahn-Jordan decomposition, 
we shall confuse notation and also call \(\mathbf{y}^{*}\)
a solution of task~(\ref{lp:CNT3}).
A solution \(\mathbf{y^{*}}\) generally does not define a probability distribution;
instead, it provides a signed \(\sigma\)-additive measure whose total variation
is smallest among all signed measures with marginals that agree both
with the bunches and multimaximal connections of the system \(\mathcal{R}_{n}\).
A solution of this task gives a true probability measure 
if and only if the system is noncontextual.

If a system \(\mathcal{R}_{n}\) is consistently connected, the contextual
fraction proposed by Abramsky et al. \cite{Abramsky.2017.Contextual}
can be computed solving the following linear programming task:
\begin{equation}
\label{lp:CNTF}
\begin{tabular}{|ccc|}
\hline
find           & maximizing                            & subject to \tabularnewline
\hline 
\(\mathbf{z}\) & \(\mathbf{1}^{\intercal} \mathbf{z}\) &
\(\mathbf{M}_{(.)} \mathbf{z} \leq \mathbf{p}_{(.)}^{*}\) \tabularnewline
& & \(\mathbf{z} \geq 0\)                                        \tabularnewline 
& & \(\mathbf{1}^{\intercal} \mathbf{z} \leq 1\)                 \tabularnewline 
\hline 
\end{tabular}.
\end{equation}
For any solution \(\mathbf{z^{*}}\), \(\text{CNTF}(\mathcal{R}_{n}) = 1 - \mathbf{1^{\intercal}z^{*}}\).
The previous task is equivalent to the
one proposed in \cite{Abramsky.2017.Contextual} which uses a simpler
representation of the system \cite{Dzhafarov.2019.Contextuality-by-Default}.
A solution \(\mathbf{z^{*}}\) generally does not define a probability distribution.
It is a defective \(\sigma\)-additive measure
with total measure \(0 \leq T \leq 1\),
that is a true probability measure 
if and only if the system is noncontextual.
Note that both tasks (\ref{lp:CNT3}) and (\ref{lp:CNTF}),
used to compute \(\text{CNT}_{3}\) and CNTF, respectively,
have in general infinitely many solutions.

Now, if we consider the consistified system \(\mathcal{R}_{n}^{\ddagger}\)
of a cyclic system \(\mathcal{R}_{n}\), then 
if \(\mathcal{R}_{n}\) is consistently connected, 
equality (\ref{eq:cntfequality}) is satisfied by Th. 7 of Ref.~\cite{Dzhafarov.2019.Contextuality-by-Default}:
\begin{align}
\text{CNTF}(\mathcal{R}_{n})    & =\text{CNTF}(\mathcal{R}_{n}^{\ddagger}). \label{eq:cntfequality}
\end{align}
Moreover, Th.~7 of Ref.~\cite{Dzhafarov.2019.Contextuality-by-Default}
also shows that, regardless of consistent connectedness, the linear
programming task to compute \(\text{CNTF}(\mathcal{R}_{n}^{\ddagger})\)
is equivalent to the task in expression (\ref{lp:CNTF}) where \(\mathbf{p}^{*}_{(.)}\) 
describes system \(\mathcal{R}_{n}^{\ddagger}\).
Hence, we will use equality (\ref{eq:cntfequality})
as the definition of the contextual fraction for inconsistently connected
systems and compute it using task~(\ref{lp:CNTF}).

\section{Relating \texorpdfstring{$\text{CNT}_3$}{} and CNTF in cyclic systems}

To relate the two measures of degree of contextuality
\(\text{CNT}_{3}\) and CNTF of a system \(\mathcal{R}_{n}\),
we consider the set of its defective quasi-couplings.
Let
\begin{equation}
    \label{eq:def_quasi_pyramid}
    \mathcal{Q}_{n} =
        \left\{
            \mathbf{x} \in \mathbb{R}^{2^{2n}} :
                \mathbf{M}_{(.)} \mathbf{x} \leq \mathbf{p}^{*}_{(.)}
                \, \text{ and } \,
                \mathbf{1}^{\intercal} \mathbf{x} \leq 1
        \right\},
\end{equation}
the convex pyramid obtained
by the intersection of the convex polyhedral cone---%
that is, a space closed under addition and multiplication by non-negative scalars generated
by the intersection of a finite number of half-spaces which have \(\mathbf{0}\) on their boundary
\cite{Lovasz.1986.Matchinga,Weyl.1952.elementary}---%
defined by the half-spaces \(\mathbf{M}_{(.)} \mathbf{x} \leq \mathbf{p}^{*}_{(.)}\)
and the half-space \(\mathbf{1}^{\intercal} \mathbf{x} \leq 1\).
Figure~\ref{fig:setC} schematically illustrates the set \(\mathcal{Q}_{n}\).
We see that the intersection of hyperplane \(\mathbf{1}^{\intercal} \mathbf{y} = 1\) 
and \(\mathcal{Q}_{n}\) defines the face of the pyramid 
on which all solutions to task~(\ref{lp:CNT3})
used to compute \(\text{CNT}_{3}\) lie.
Similarly, the intersection of hyperplane \(\mathbf{1}^{\intercal} \mathbf{z} = 1 - \text{CNTF}\),
\(\mathcal{Q}_{n}\), and the nonnegative orthant of \(\mathbb{R}^{2^{2n}}\),
defines a slice on whose surface lie all solutions to task~(\ref{lp:CNTF})
used to compute CNTF.

\begin{figure}[ht]
    \centering
    \includegraphics[width=0.8\textwidth]{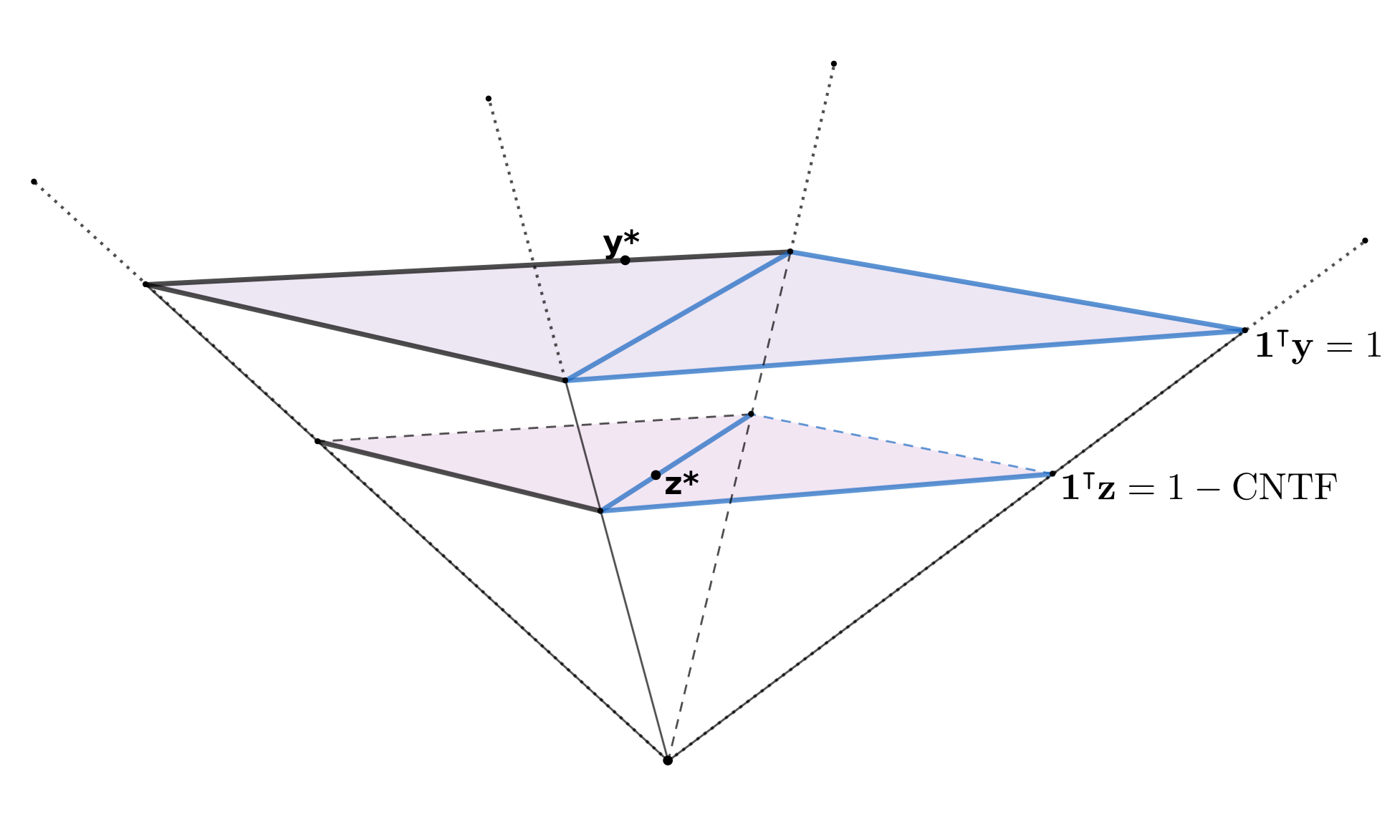}
    \caption{Scheme of the pyramid of defective quasi-couplings \(\mathcal{Q}_{n}\).
    The intersection of \(\mathcal{Q}_{n}\) and the nonnegative orthant of \(\mathbb{R}^{2^{2n}}\) is
    illustrated via the blue lines on the two depicted slices cutting through \(\mathcal{Q}_{n}\).
    Quasi-couplings \(\mathbf{y}^{*}\) lie on the slice \(\mathbf{1}^{\intercal} \mathbf{y} = 1\)
    and defective couplings \(\mathbf{z}^{*}\) that are solutions to task~(\ref{lp:CNTF})
    lie within the closed region delimited by blue edges on the slice
    \(\mathbf{1}^{\intercal} \mathbf{z} = 1 - \text{CNTF}\).
    }
    \label{fig:setC}
\end{figure}

\begin{lemma}
    \label{lem:single_neg}
    If a cyclic system \(\mathcal{R}_{n}\) is contextual, 
    there exists some solution \(\mathbf{y}^{*}\) of task~(\ref{lp:CNT3}) with 
    a single negative component.
\end{lemma}
\begin{proof}
    Fix \(i \in \{1, \ldots, n\},\) and choose an event 
    \(S = \left\{ S_{i}^{i} = 1, S_{i}^{i \ominus 1} = 0 \right\}\)
    such that a multimaximal coupling of \(\mathcal{R}_{i}\)
    has, without loss of generality,
    \(
    \Pr \left(
    T_{i}^{i} = 1, T_{i}^{i \ominus 1} = 0 
    \right) = 0.
    \)\footnote{%
    If for no \(\mathcal{R}_{i}\),
    \(
    \Pr \left(
    T_{i}^{i} = 1, T_{i}^{i \ominus 1} = 0 
    \right) = 0
    \), replace \(R_{i}^{c}\) in the system with \(1 - R_{i}^{c}\) for some \(i\).
    }
    Look at the row \(u\) of \(\mathbf{M}_{(.)}\) corresponding to
    \(
    \Pr \left(
    T_{i}^{i} = 1, T_{i}^{i \ominus 1} = 0 
    \right) = 0
    \) and let
    \(V\) be the set of indices \(j \in \{1, \ldots, 2^{2n}\}\) such that \(\mathbf{M}_{(.), u, j} = 1\).
    Choose any \(v \in V\), and let \(s_{v}\) be the \(v\)th component of \(S\).
    Define \(\mathbf{q}_{(.)}^{*}\) component-wise by taking
    \(\mathbf{q}^{*}_{(.),i} = \mathbf{p}^{*}_{(.),i} + \frac{1}{2} \text{CNT}_{3}\) if
    the event \(s'\) whose probability is the \(i\)th component of \(\mathbf{p}_{(.)}^{*}\) is contained in \(s_{v}\),
    and \(\mathbf{q}^{*}_{(.),i} = \mathbf{p}^{*}_{(.),i}\), otherwise.
    Lastly, let
\begin{equation}
    \mathcal{H}_{v} = \left\{ \mathbf{x} \in \mathbb{R}^{2^{2n}} : 
        \mathbf{1}^{\intercal} (\mathbf{x} - \mathbf{e}_{v}) = 1 + \frac{1}{2} \text{CNT}_{3} 
        \text{ and }
        \mathbf{M}_{(.)} (\mathbf{x} - \mathbf{e}_{v}) = 
        \mathbf{q}_{(.)}^{*} 
        \right\},
\end{equation}
    where \(\mathbf{e}_{v}\) is the unit vector with a \(1\) on its \(v\)th component,
    and choose a point \(\mathbf{w}^{*}\) with zero \(v\)th component
    in the intersection of \(\mathcal{H}\)  and the nonnegative orhtant of \(\mathbb{R}^{2^{2n}}\).
    Clearly, the point \(\mathbf{y}^{*} = \mathbf{w}^{*} - \frac{1}{2} \text{CNT}_{3} \mathbf{e}_{v}\)
    is a solution of task~(\ref{lp:CNT3})
    with \(\mathbf{y}^{*}_{v} = -\frac{1}{2} \text{CNT}_{3}\) its sole negative component.
\end{proof}

\begin{lemma}
   \label{lem:equidistant_z}
   Let \(\mathcal{R}_{n}\) be a contextual cyclic system.
   Given a solution \(\mathbf{y}^{*}\) of task (\ref{lp:CNT3})
   as in Lemma \ref{lem:single_neg},
   a solution \(\mathbf{z}^{*}\) of task~(\ref{lp:CNTF})
   can be constructed such that 
   \(|\mathbf{y}^{*}_i| \geq \mathbf{z}^{*}_i\), \(i = 1, \ldots, 2^{2n}\),
   and
   \(||\mathbf{y}^{*} - \mathbf{z}^{*}||_1 = n\text{CNT}_{3}\).
\end{lemma}

\begin{proof}
 
Choose a solution \(\mathbf{y}^{*}\) in accordance to Lemma~\ref{lem:single_neg}.
Let \(\hat{\mathbf{x}}_{1} = \mathbf{y}^{*}_{v}\mathbf{e}_{v}\)
where \(v\) is the index of the only negative component of \(\mathbf{y}^{*}\).
Using this \(v\), let \(s_{v}\) and \(\mathbf{q}_{(.)}^{*}\) be defined as in Lemma~\ref{lem:single_neg},
and let \(U\) be the set of indices \(u \in \{1, \ldots, 12n\}\) such that
\(\mathbf{M}_{(.),u, v} = 1\).
Note that \(|U| = 4n\), where there are \(n\) indices such that
\(\mathbf{p}_{(.), u}\) corresponds to 
\(\Pr(R_{i}^{i} = r_{i}^{i}, R_{i \oplus i}^{i} = r_{i \oplus 1}^{i})\),
one for each of the \(n\) contexts of \(\mathcal{R}_{n}\);
another \(n\) correspond to
\(\Pr(T_{i}^{i} = r_{i}^{i}, T_{i}^{i \ominus 1} = r_{i}^{i \ominus 1})\),
one per content;
and \(2n\) correspond to one probability
\(\Pr(R_{j}^{i} = r_{j}^{i})\) for each random variable in the system.

Let \(\mathbf{M}_{U}\) be the submatrix of \(\mathbf{M}_{(.)}\)
whose rows are indexed by \(U\),
and \(\mathbf{M}_{U'}\) the matrix with the remaining
rows of \(\mathbf{M}_{(.)}\).
(Note that matrix \(\mathbf{M}_{U}\) is a reduction of
matrix \(\mathbf{M}_{(.)}\) in the same manner as
\(\mathbf{M}\), with the event \(s_{v}\) taking the place of the
event \(\left\{S_{i}^{i} = 1, S_{i \oplus 1}^{i} = 1 \right\}_{i = 1, \ldots, n}\)
for its construction, see Ref.~\cite{Dzhafarov.2016.Context-content}.)
Define \(\mathbf{p}_{U}^{*}\) and \(\mathbf{p}_{U'}^{*}\) analogously.
We can then rewrite \(\mathcal{Q}_{n}\) as the intersection of
\begin{equation}
    \mathcal{Q}_{U} =
        \left\{
            \mathbf{x} \in \mathbb{R}^{2^{2n}} :
                \mathbf{M}_{U} \mathbf{x} \leq \mathbf{p}^{*}_{U}
                \, \text{ and } \,
                \mathbf{1}^{\intercal} \mathbf{x} \leq 1
        \right\},
\end{equation}
and
\begin{equation}
    \mathcal{Q}_{U'} =
        \left\{
            \mathbf{x} \in \mathbb{R}^{2^{2n}} :
                \mathbf{M}_{U'} \mathbf{x} \leq \mathbf{p}^{*}_{U'}
                \, \text{ and } \,
                \mathbf{1}^{\intercal} \mathbf{x} \leq 1
        \right\}.
\end{equation}
From the definition of \(\mathbf{M}\) (see Ref.\cite{Dzhafarov.2016.Context-content}),
we have that the dimension of \(\mathcal{Q}_{n}\) is \(4n +1\).
Similarly, the dimension of \(\mathcal{Q}_{U}\) is \(4n + 1\) because it is
constructed by a minimal subset of defining inequalities of \(\mathcal{Q}_{n}\).

Define \(\mathbf{w}^{*} = \mathbf{y}^{*} - \hat{\mathbf{x}}_{1}\).
Since \(\mathbf{M}_{(.)}\mathbf{w}^{*} = \mathbf{q}_{(.)}^{*}\),
\(\mathbf{w}^{*} \notin \mathcal{Q}_{n}\).
Clearly, \(\mathbf{w}^{*} \notin \mathcal{Q}_{U}\) and \(\mathbf{w}^{*} \in \mathcal{Q}_{U'}\).
Let us next consider the task 
\begin{equation}
\label{lp:y_to_z_full}
\begin{tabular}{|ccc|}
\hline
find           &  minimizing                            & subject to \tabularnewline
\hline 
\(\mathbf{x}\) &  \(\mathbf{1}^{\intercal} \mathbf{x}\) 
& \(\mathbf{M}_{U} (\mathbf{w}^{*} - \mathbf{x}) \leq \mathbf{p}_{U}^{*}\) \tabularnewline 
& & \(\mathbf{M}_{U'} (\mathbf{w}^{*} - \mathbf{x}) \leq \mathbf{p}_{U'}^{*}\) \tabularnewline
& & \(\mathbf{x} \geq 0\)                                           \tabularnewline 
& & \(\mathbf{1}^{\intercal} (\mathbf{w}^{*} - \mathbf{x}) \leq 1\) \tabularnewline 
& & \(\mathbf{e}_{v}^{\intercal} \mathbf{x} = 0\)                   \tabularnewline 
\hline 
\end{tabular}.
\end{equation}
This task must have a solution, since \(\mathbf{x}^{*} = \mathbf{w}^{*}\)
satisfies all its constraints.
Additionally, it is evident that the second set of restrictions
(those associated with \(\mathbf{p}_{U'}^{*}\))
place no restriction to finding the solution because, by construction,
\(\mathbf{M}_{U'} \mathbf{w}^{*} = \mathbf{p}_{U'}^{*}\);
hence, any vector \(\mathbf{x} \geq 0\) will satisfy that set of inequalities.
Further examination of the constraints shows immediately that
for any solution \(\mathbf{x}^{*}\),
\(\mathbf{1}^{\intercal} \mathbf{x}^{*} \geq -2 \mathbf{y}^{*}_{v}\).
Similarly, inspecting the constraints associated with \(\mathbf{p}_{U}^{*}\)
reveals that whenever a vector \(\mathbf{x'}\) satisfies
\(\mathbf{M}_{U, u}^{\intercal} (\mathbf{w}^{*} - \mathbf{x'}) \leq {\mathbf{p}^{*}_{U}}_{u}\),
where \({\mathbf{p}^{*}_{U}}_{u}\) is a probability 
\(\Pr(R_{i}^{i} = r_{i}^{i}, R_{i \oplus i}^{i} = r_{i \oplus 1}^{i})\),
then
\(\mathbf{M}_{U,t}^{\intercal} (\mathbf{w}^{*} - \mathbf{x'}) \leq {\mathbf{p}^{*}_{U}}_{t}\),
where \({\mathbf{p}^{*}_{U}}_{t}\) is a probability 
\(\Pr(R_{i}^{i} = r_{i}^{i})\) or
\(\Pr(R_{i \oplus 1 }^{i} = r_{i \oplus 1}^{i})\)%
---for the same \(i\) in the event corresponding to
\(\Pr(R_{i}^{i} = r_{i}^{i}, R_{i \oplus i}^{i} = r_{i \oplus 1}^{i})\)---,%
will also be satisfied.
An analogous observation can be made when
\({\mathbf{p}^{*}_{U}}_{u}\) is a probability 
\(\Pr(T_{i}^{i} = r_{i}^{i}, T_{i}^{i \ominus 1} = r_{i}^{i \ominus 1})\).
Therefore, at most \(2n\) of the constraints imposed
via matrix \(\mathbf{M}_{U}\) are active in determining the solution space
of task~(\ref{lp:y_to_z_full})

Let \(\mathbf{M}_{w}\) and \(\mathbf{p}_{w}^{*}\) contain the rows and probabilities
of \(\mathbf{M}_{U}\) and \(\mathbf{p}_{U}^{*}\), respectively,
corresponding to bunch and connection probabilities.
Since the rows of \(\mathbf{M}_{U}\) are linearly independent,
so are the rows of \(\mathbf{M}_{w}\), and the latter
has full row rank \(2n\).
Given the considerations in the above paragraph,
task~(\ref{lp:y_to_z_full}) is equivalent to task
\begin{equation}
\label{lp:y_to_z_short}
\begin{tabular}{|ccc|}
\hline
find           &  minimizing                            & subject to \tabularnewline
\hline 
\(\mathbf{x}\) &  \(\mathbf{1}^{\intercal} \mathbf{x}\) 
& \(\mathbf{M}_{w} (\mathbf{w}^{*} - \mathbf{x}) \leq \mathbf{p}_{w}^{*}\) \tabularnewline 
& & \(\mathbf{x} \geq 0\)                                           \tabularnewline 
& & \(\mathbf{e}_{v}^{\intercal} \mathbf{x} = 0\)                   \tabularnewline 
\hline 
\end{tabular}.
\end{equation}
Now, the constraint \(\mathbf{e}_{v}^{\intercal} \mathbf{x} = 0\)
can be replaced by a modification of column \(v\) of matrix \(\mathbf{M}_{w}\)
in which the column is replaced by a vector of zeros.
This effectively reduces its rank to \(2n - 1\).
We further note that, in standard form, the constraints for task~(\ref{lp:y_to_z_short}) are
\(\mathbf{M}_{w} \mathbf{x} \geq \mathbf{M}_{w} \mathbf{w}^{*} - \mathbf{p}_{w}^{*}\),
and the deficiency in rank just introduced implies that there is some
row of \(\mathbf{M}_{w}\) that may be safely removed for purposes of 
finding a solution \(\mathbf{x}^{*}\).
Since (assuming the modified matrix)
\(\mathbf{M}_{w} \mathbf{x} \geq \mathbf{M}_{w} \mathbf{w}^{*} - \mathbf{p}_{w}^{*}\)
is an underdetermined system with \(2n - 1\) inequalities,
there exists a solution \(\mathbf{x}^{*}\) such that
all components of \(\mathbf{x}^{*}\) but \(2n - 1\) are zero.
Therefore, we see that for a solution \(\mathbf{x}^{*}\),
\(\mathbf{M}_{w} \mathbf{x}^{*} = \mathbf{M}_{w} \mathbf{w}^{*} - \mathbf{p}_{w}^{*}\),
and that
\(\mathbf{1}^{\intercal} \mathbf{x}^{*} = \mathbf{1}^{\intercal} (\mathbf{M}_{w} \mathbf{w}^{*} - \mathbf{p}_{w}^{*})
= -(2n - 1) \mathbf{y}^{*}_{v}\). 
The statement is obtained by noting that 
task~(\ref{lp:y_to_z_full}) is equivalent to maximizing 
\(\mathbf{1}^{\intercal} (\mathbf{w}^{*} - \mathbf{x})\) 
under the same constraints, 
which is essentially task~(\ref{lp:CNTF}).
In other words,
\(\mathbf{z}^{*} \equiv \mathbf{y}^{*} - \mathbf{\hat{x}}_{1} - \mathbf{x}^{*}\) 
is an optimal solution to task~(\ref{lp:CNTF}).

\end{proof}

\begin{lemma}
    \label{lem:diff_y_x}
    \(\left\Vert \mathbf{y}^{*} - \mathbf{z}^{*} \right\Vert_{1} = \text{CNTF} + \text{CNT}_{3}\)
\end{lemma}
\begin{proof}
Choose solutions \(\mathbf{y}^{*}\) and \(\mathbf{z}^{*}\)
in accordance to Lemmas~\ref{lem:single_neg} and \ref{lem:equidistant_z}.
Then 
\begin{align*}
\left\Vert \mathbf{y}^{*} \right\Vert_{1} 
    &=  \left\Vert \mathbf{y}^{*} - \mathbf{z}^{*} + \mathbf{z}^{*} \right\Vert_{1} \\
    &= \left\Vert \mathbf{y}^{*} - \mathbf{z}^{*} \right\Vert_{1} + \left\Vert \mathbf{z}^{*} \right\Vert_{1} \\
    &= \left\Vert \mathbf{y}^{*} - \mathbf{z}^{*} \right\Vert_{1} + 1 - \text{CNTF} 
\end{align*}
where the second line follows by the choice of \(\mathbf{y}^{*}\) and \(\mathbf{z}^{*}\).
The statement follows immediately by noting that
\(
\left\Vert \mathbf{y}^{*} \right\Vert_{1} = 1 + \text{CNT}_{3}
\).
\end{proof}

\begin{theorem}
    If \(\mathcal{R}_{n}\) is a cyclic system of rank \(n\), 
    then \(\text{CNTF}(\mathcal{R}_{n}) = (n - 1) \text{CNT}_{3}(\mathcal{R}_{n})\)
\end{theorem}

\begin{proof}
The relation in the statement is trivially true for any noncontextual system;
hence, assume that \(\mathcal{R}_{n}\) is a contextual cyclic system of rank \(n\).
Choose solutions \(\mathbf{y}^{*}\) and \(\mathbf{z}^{*}\)
in accordance to Lemmas~\ref{lem:single_neg} and \ref{lem:equidistant_z}.
By Lemma~\ref{lem:equidistant_z},
\begin{align*}
\left\Vert \mathbf{y}^{*} - \mathbf{z^*}\right\Vert_{1} 
  = n \text{CNT}_{3}.
\end{align*}
And from Lemma~\ref{lem:diff_y_x}, it follows that
\begin{equation}
\text{CNTF} = (n - 1) \text{CNT}_{3}.
\end{equation}
\end{proof}

\begin{corollary}
\(\text{CNT}_{3}\) is not invariant to consistification.
If \(\mathcal{R}_{n}\) is a contextual cyclic system and
\(\mathcal{R}_{n}^{\ddagger}\) is its consistification,
then
\[
    \text{CNT}_{3}(\mathcal{R}_{n}^{\ddagger}) =
        \frac{n - 1}{2n - 1}
        \text{CNT}_{3}(\mathcal{R}_{n}).
\]
\end{corollary}

\begin{proof}
\(
(n - 1) \text{CNT}_{3}(\mathcal{R}_{n}) = 
    \text{CNTF}(\mathcal{R}_{n}) = 
    \text{CNTF}(\mathcal{R}_{n}^{\ddagger}) = 
    (2n - 1) \text{CNT}_{3}(\mathcal{R}_{n}^{\ddagger}).
\)
\end{proof}

\begin{example}[Consistently connected system]
\label{ex:cons_connected}
Consider a cyclic system \(\mathcal{R}_{3}\) with bunch joint distributions 
\begin{equation}
    \label{ex:bunches}
    \def\arraystretch{1.5}
    \begin{array}{c|c|c}
        & R_{i}^{i} = 0 & R_{i}^{i} = 1 \\
        \hline
        R_{i\oplus 1}^{i} = 0 & \nicefrac{1}{8} & \nicefrac{3}{8} \\
        \hline
        R_{i\oplus 1}^{i} = 1 & \nicefrac{3}{8} & \nicefrac{1}{8}
    \end{array}\ .
\end{equation}
The system \(\mathcal{R}_{3}\) is consistently connected and can be represented by the vector:
\[
    \mathbf{p^*}^\intercal = (
        \nicefrac{1}{2},
        \nicefrac{1}{2},
        \nicefrac{1}{2},
        \nicefrac{1}{2},
        \nicefrac{1}{2},
        \nicefrac{1}{2},
        \nicefrac{1}{8},
        \nicefrac{1}{8},
        \nicefrac{1}{8},
        \nicefrac{1}{2},
        \nicefrac{1}{2},
        \nicefrac{1}{2}
        ).
\]
Let $\{\mathbf{e}_{j}\}_{j = 1, \ldots, 64}$, be the standard basis of $\mathbb{R}^{64}$, 
we can write a solution to task (\ref{lp:CNT3}) with a single negative mass (as in Lemma~\ref{lem:single_neg})
\[
\mathbf{y}^{*} = 
    \frac{1}{16} \left(
    3 \mathbf{e}_{7} -
      \mathbf{e}_{14} +
    2 \mathbf{e}_{25} +
      \mathbf{e}_{26} +
    3 \mathbf{e}_{31} +
    2 \mathbf{e}_{34} +
      \mathbf{e}_{38} +
    2 \mathbf{e}_{40} +
      \mathbf{e}_{42} +
    2 \mathbf{e}_{58}
    \right).
\]

To highlight the dimension of the solution space \(\mathcal{Q}_{U}\),
this solution can be further re-expressed as a linear combination of 
the following \(\text{L}_{1}\)-orthonormal vectors
\(\{\hat{\mathbf{x}}_{j}\}_{j = 1, \ldots, 6}\):
\begin{equation*}
    \begin{aligned}
    \hat{\mathbf{x}}_{1} &= -\mathbf{e}_{14}, \qquad&
    \hat{\mathbf{x}}_{4} &= (\mathbf{e}_{26} + \mathbf{e}_{38} + \mathbf{e}_{42}) / 3, \\
    \hat{\mathbf{x}}_{2} &=  \mathbf{e}_{7}, \qquad&
    \hat{\mathbf{x}}_{5} &= \mathbf{e}_{31}, \\
    \hat{\mathbf{x}}_{3} &= (\mathbf{e}_{25} + \mathbf{e}_{40}) / 2, \qquad&
    \hat{\mathbf{x}}_{6} &= (\mathbf{e}_{34} + \mathbf{e}_{58}) / 2. 
    \end{aligned}
\end{equation*}
In terms of these vectors, we have
\[
\mathbf{y}^{*} = 
    \frac{1}{16} \left(
      \hat{\mathbf{x}}_{1} +
    3 \hat{\mathbf{x}}_{2} +
    4 \hat{\mathbf{x}}_{3} +
    3 \hat{\mathbf{x}}_{4} +
    3 \hat{\mathbf{x}}_{5} +
    4 \hat{\mathbf{x}}_{6} 
    \right).
\]

Now, we can use the construction in Lemma~\ref{lem:equidistant_z} to find the point
\[
\mathbf{z}^{*} = 
    \frac{1}{16} \left(
    0 \hat{\mathbf{x}}_{1} +
    2 \hat{\mathbf{x}}_{2} +
    3 \hat{\mathbf{x}}_{3} +
    2 \hat{\mathbf{x}}_{4} +
    2 \hat{\mathbf{x}}_{5} +
    3 \hat{\mathbf{x}}_{6} 
    \right).
\]
which is a solution to task (\ref{lp:CNTF}) to compute \(\text{CNTF}\).
For this system 
\(\text{CNT}_{3} = \nicefrac{1}{8}\) and 
\[\text{CNTF} = \frac{1}{4} = (3 - 1) \text{CNT}_{3}.\]
\end{example}

\begin{example}[Inconsistently connected system]
\label{ex:incons_connected}
    Consider the system \(\mathcal{R}'_{3}\) in which the distribution of the third bunch 
    of system \(\mathcal{R}_{3}\) from example~\ref{ex:cons_connected}
    is replaced by
    \begin{equation}
    \def\arraystretch{1.5}
    \begin{array}{c|c|c}
        & R_{3}^{3} = 0 & R_{3}^{3} = 1 \\
        \hline
        R_{1}^{3}=0 & \nicefrac{1}{8} & \nicefrac{7}{16} \\
        \hline
        R_{1}^{3}=1 & \nicefrac{3}{8} & \nicefrac{1}{16}
    \end{array}\ .
\end{equation}
The system \(\mathcal{R}'_{3}\) is inconsistently connected and can be represented by the vector:
\[
    \mathbf{p^*}^\intercal = (
        \nicefrac{1}{2},
        \nicefrac{1}{2},
        \nicefrac{1}{2},
        \nicefrac{1}{2},
        \nicefrac{\mathbf{7}}{\mathbf{16}},
        \nicefrac{1}{2},
        \nicefrac{1}{8},
        \nicefrac{1}{8},
        \nicefrac{\mathbf{1}}{\mathbf{16}},
        \nicefrac{1}{2},
        \nicefrac{1}{2},
        \nicefrac{\mathbf{7}}{\mathbf{16}}
        ).
\]
One possible solution \(\mathbf{y}^{*}\) of task (\ref{lp:CNT3}) 
for system \(\mathcal{R}'_{3}\) can be written as a linear combination of 
the following \(\text{L}_{1}\)-orthonormal vectors
\(\{\hat{\mathbf{x}}_{j}\}_{j = 1, \ldots, 6}\):
\begin{equation*}
    \begin{aligned}
    \hat{\mathbf{x}}_{1} &= -\mathbf{e}_{49}, \qquad&
    \hat{\mathbf{x}}_{4} &= (\mathbf{e}_{23} + \mathbf{e}_{39}) / 2, \\
    \hat{\mathbf{x}}_{2} &= (\mathbf{e}_{7} + \mathbf{e}_{31} + \mathbf{e}_{40}) / 3, \qquad&
    \hat{\mathbf{x}}_{5} &=  \mathbf{e}_{25}, \\
    \hat{\mathbf{x}}_{3} &=  \mathbf{e}_{27}, \qquad&
    \hat{\mathbf{x}}_{6} &= (\mathbf{e}_{34} + \mathbf{e}_{58}) / 2. 
    \end{aligned}
\end{equation*}
with
\[
\mathbf{y}^{*} = 
    \frac{1}{16} \left(
      \hat{\mathbf{x}}_{1} +
    6 \hat{\mathbf{x}}_{2} +
      \hat{\mathbf{x}}_{3} +
    2 \hat{\mathbf{x}}_{4} +
    2 \hat{\mathbf{x}}_{5} +
    6 \hat{\mathbf{x}}_{6} 
    \right).
\]

Similarly to the previous example, use the construction in Lemma~\ref{lem:equidistant_z} to find the point
\[
\mathbf{z}^{*} = 
    \frac{1}{16} \left(
    0 \hat{\mathbf{x}}_{1} +
    6 \hat{\mathbf{x}}_{2} +
      \hat{\mathbf{x}}_{3} +
    0 \hat{\mathbf{x}}_{4} +
      \hat{\mathbf{x}}_{5} +
    4 \hat{\mathbf{x}}_{6} 
    \right).
\]
which is a solution to task (\ref{lp:CNTF}) to compute \(\text{CNTF}\).
For this system 
\(\text{CNT}_{3} = \nicefrac{1}{8}\) and 
\[\text{CNTF} = \frac{1}{4} = (3 - 1) \text{CNT}_{3}.\]
\end{example}

\section{Discussion}

The result presented in this paper proves that the conjecture in \cite{Kujala.2019.Measures}
is indeed true.
Moreover, we can now affirm that all the fundamentally different approaches 
to quantify contextuality currently found in the literature are
proportional to each other within the class of cyclic systems.
The proportionality relations among the measures are:
\begin{equation}
      2 \text{CNT}_{0} =
      2 \text{CNT}_{1} =
      2 \text{CNT}_{2} =
        \text{CNTF}    =
(n - 1) \text{CNT}_{3}.
\end{equation}
The equality of the first three measures was shown in
Ref.~\cite{Dzhafarov.2020.Contextuality},
the third equality was proved in Ref.~\cite{Cervantes.2023.note},
and the last equality, in this paper. 
It should also be noticed that the hierarchical measure of contextuality 
proposed in Ref.~\cite{Cervantes.2020.Contextuality} reduces to 
$\text{CNT}_{2}$ for cyclic systems;
therefore, it also satisfies the proportionality to the other measures.
This result therefore completes the contextuality theory of cyclic systems
and its measures.

 However, as noted in 
Refs.\cite{Dzhafarov.2020.Contextuality,cervantes2023hypercyclic,Cervantes.2023.note},
 the relations among these measures are not as simple as
 in other types of systems of random variables.
 In Ref.~\cite{Dzhafarov.2020.Contextuality}, one can find examples 
 of non-cyclic systems for which
 \(\text{CNT}_{1}\) and \(\text{CNT}_{2}\) are not functions of each other.
 A class of examples is considered in Ref.~\cite{Cervantes.2023.note,Cervantes.2021.Contextuality}
 to show the same lack of functional relation
 between \(\text{CNT}_{2}\) and CNTF outside of cyclic systems.
 Lastly, in Ref.~\cite{cervantes2023hypercyclic}, some examples
 of hypercyclic systems of order higher than \(2\)---%
 cyclic systems are a special case of this class where order equals \(2\)---%
 that show that in general there is no functional relation
 among any of the measures here considered.

\paragraph{Acknowledgments}

The authors would like to thank Ehtibar N.\ Dzhafarov, Alisson Tezzin and Bárbara Amaral for helpful discussions. 
GC worked under the Financial Support of the 
Coordenação de Aperfeiçoamento de Pessoal de Nível Superior (CAPES) - Programa de Excelência Acadêmica (PROEX) - Brasil and 
the Conselho Nacional de Desenvolvimento Científico e Tecnológico (CNPq) - Brasil.


\begin{thebibliography}{9}

\bibitem{Fine.1982.Joint}
Fine A. 1982  Joint distributions, quantum correlations, and commuting observables. {\em J. Math. Phys.} \textbf{23}, 1306--1310.
(\href{http://dx.doi.org/10.1063/1.525514}{10.1063/1.525514})

\bibitem{Clauser.1969.Proposed}
Clauser JF, Horne MA, Shimony A, Holt RA. 1969  Proposed experiment to test local hidden-variable theories. {\em Phys. Rev. Lett.} \textbf{23}, 880--884.
(\href{http://dx.doi.org/10.1103/PhysRevLett.23.880}{10.1103/PhysRevLett.23.880})

\bibitem{Klyachko.2008.Simple}
Klyachko AA, Can MA, Binicio{\u{g}}lu S, Shumovsky AS. 2008  Simple test for hidden variables in spin-1 systems. {\em Phys. Rev. Lett.} \textbf{101}, 020403.

\bibitem{Kujala.2019.Measures}
Kujala JV, Dzhafarov EN. 2019  Measures of contextuality and non-contextuality. {\em Philos. Trans. R. Soc. A} \textbf{377}, 20190149.
(\href{http://dx.doi.org/10.1098/rsta.2019.0149}{10.1098/rsta.2019.0149})

\bibitem{Cervantes.2023.note}
Cervantes VH. 2023  A note on the relation between the {{Contextual Fraction}} and {{\(\text{CNT}_{2}\)}}. {\em J. Math. Psychol.} \textbf{112}, 102726.
(\href{http://dx.doi.org/10.1016/j.jmp.2022.102726}{10.1016/j.jmp.2022.102726})

\bibitem{Dzhafarov.2019.Contextuality-by-Default}
Dzhafarov EN. 2023  The {{Contextuality-by-Default}} view of the {{Sheaf-Theoretic}} approach to contextuality. In Palmigiano A, Sadrzadeh M, editors, {\em Samson {Abramsky} on Logic and Structure in Computer Science and Beyond} ,  Outstanding Contributions to Logic. {Dordrecht}: {Springer Nature}.
(\href{http://dx.doi.org/10.1007/978-3-031-24117-8}{10.1007/978-3-031-24117-8})

\bibitem{Dzhafarov.2016.Contextuality-by-Default}
Dzhafarov EN, Kujala JV, Cervantes VH. 2016  Contextuality-by-{{Default}}: A Brief Overview of Ideas, Concepts, and Terminology. In Atmanspacher H, Filk T, Pothos E, editors, {\em Quantum {{Interaction}}}, Lecture {{Notes}} in {{Computer Science}},  vol. 9535,  pp. 12--23. {Dordrecht}: {Springer}.

\bibitem{Dzhafarov.2017.Contextuality-by-Default}
Dzhafarov EN, Kujala JV. 2017  Contextuality-by-{{Default}} 2.0: Systems with Binary Random Variables. In {\em Quantum {{Interaction}}}, Lecture {{Notes}} in {{Computer Science}},  vol. 10106,  pp. 16--32. {Dordrecht}: {Springer}.

\bibitem{Dzhafarov.2016.Context-content}
Dzhafarov EN, Kujala JV. 2016  Context-Content Systems of Random Variables: The {{Contextuality}}-by-{{Default}} Theory. {\em J. Math. Psychol.} \textbf{74}, 11--33.
(\href{http://dx.doi.org/10.1016/j.jmp.2016.04.010}{10.1016/j.jmp.2016.04.010})

\bibitem{Dzhafarov.2020.Contextuality}
Dzhafarov EN, Kujala JV, Cervantes VH. 2020  Contextuality and noncontextuality measures and generalized {{Bell}} inequalities for cyclic systems. {\em Phys. Rev. A} \textbf{101}, 042119.
(Erratum Note 1 in Phys. Rev. A 101:069902, 2020.) (Erratum Note 2 in Phys. Rev. A 103:059901, 2021.) (\href{http://dx.doi.org/10.1103/PhysRevA.101.042119}{10.1103/PhysRevA.101.042119})

\bibitem{Bell.1964.Einstein}
Bell JS. 1964  On the {Einstein}-{Podolsky}-{Rosen} Paradox. {\em Physics} \textbf{1}, 195--200.

\bibitem{Bell.1966.problem}
Bell JS. 1966  On the Problem of Hidden Variables in Quantum Mechanics. {\em Rev. Mod. Phys.} \textbf{38}, 447--452.
(\href{http://dx.doi.org/10.1103/RevModPhys.38.447}{10.1103/RevModPhys.38.447})

\bibitem{Leggett.1985.Quantum}
Leggett AJ, Garg A. 1985  Quantum Mechanics versus Macroscopic Realism: Is the Flux There When Nobody Looks?. {\em Phys. Rev. Lett.} \textbf{54}, 857--860.
(\href{http://dx.doi.org/10.1103/PhysRevLett.54.857}{10.1103/PhysRevLett.54.857})

\bibitem{Suppes.1981.When}
Suppes P, Zanotti M. 1981  When Are Probabilistic Explanations Possible?. {\em Synthèse} \textbf{48}, 191--199.
(\href{http://dx.doi.org/10.1007/978-94-017-2300-8\_11}{10.1007/978-94-017-2300-8\_11})

\bibitem{Araujo.2013.All}
Araújo M, Quintino MT, Budroni C, Cunha MT, Cabello A. 2013  All Noncontextuality Inequalities for the n -Cycle Scenario. {\em Phys. Rev. A} \textbf{88}.
(\href{http://dx.doi.org/10.1103/PhysRevA.88.022118}{10.1103/PhysRevA.88.022118})

\bibitem{Kujala.2016.Proof}
Kujala JV, Dzhafarov EN. 2016  Proof of a Conjecture on Contextuality in Cyclic Systems with Binary Variables. {\em Found. Phys.} \textbf{46}, 282--299.
(\href{http://dx.doi.org/10.1007/s10701-015-9964-8}{10.1007/s10701-015-9964-8})

\bibitem{Adenier.2017.Test}
Adenier G, Khrennikov AY. 2017  Test of the no-signaling principle in the {{Hensen}} loophole-free {{CHSH}} experiment: Test of the no-signaling principle. {\em Fortschr. Phys.} \textbf{65}, 1600096.
(\href{http://dx.doi.org/10.1002/prop.201600096}{10.1002/prop.201600096})

\bibitem{Asano.2014.Violation}
Asano M, Hashimoto T, Khrennikov AY, Ohya M, Tanaka Y. 2014  Violation of Contextual Generalization of the {{Leggett}}-{{Garg}} Inequality for Recognition of Ambiguous Figures. {\em Phys. Scr.} \textbf{T163}, 014006.
(\href{http://dx.doi.org/10.1088/0031-8949/2014/T163/014006}{10.1088/0031-8949/2014/T163/014006})

\bibitem{Cervantes.2018.Snow}
Cervantes VH, Dzhafarov EN. 2018  Snow Queen Is Evil and Beautiful: Experimental Evidence for Probabilistic Contextuality in Human Choices. {\em Decision} \textbf{5}, 193--204.
(\href{http://dx.doi.org/10.1037/dec0000095}{10.1037/dec0000095})

\bibitem{Dzhafarov.2015.there}
Dzhafarov EN, Zhang R, Kujala JV. 2015  Is There Contextuality in Behavioural and Social Systems?. {\em Philos. Trans. R. Soc. A} \textbf{374}, 20150099.
(\href{http://dx.doi.org/10.1098/rsta.2015.0099}{10.1098/rsta.2015.0099})

\bibitem{Hensen.2015.Loophole-free}
Hensen B, Bernien H, Dr\'eau AE, Reiserer A, Kalb N, Blok MS, Ruitenberg J, Vermeulen RFL, Schouten RN, Abellan C, Amaya W, Pruneri V, Mitchell MW, Markham M, Twitchen DJ, Elkouss D, Wehner S, Taminiau TH, Hanson R. 2015  Loophole-Free {{Bell}} Inequality Violation Using Electron Spins Separated by 1.3 Kilometres. {\em Nature} \textbf{526}, 682--686.
(\href{http://dx.doi.org/10.1038/nature15759}{10.1038/nature15759})

\bibitem{Wang.2013.Quantum}
Wang Z, Busemeyer JR. 2013  A Quantum Question Order Model Supported by Empirical Tests of an a Priori and Precise Prediction. {\em Top. Cogn. Sci.} \textbf{5}, 689--710.
(\href{http://dx.doi.org/10.1111/tops.12040}{10.1111/tops.12040})

\bibitem{Abramsky2013}
Abramsky S. 2013  Relational databases and {{Bell}}'s theorem. In Tannen V, Wong L, Libkin L, Fan W, Tan WC, Fourman M, editors, {\em In Search of Elegance in the Theory and Practice of Computation: Essays Dedicated to Peter Buneman} ,  Lecture {{Notes}} in {{Computer Science}} pp. 13--35. {Berlin, Heidelberg}: {Springer}.
(\href{http://dx.doi.org/10.1007/978-3-642-41660-6\_2}{10.1007/978-3-642-41660-6\_2})

\bibitem{Abramsky.2017.Contextual}
Abramsky S, Barbosa RS, Mansfield S. 2017  Contextual {{Fraction}} as a measure of contextuality. {\em Phys. Rev. Lett.} \textbf{119}, 050504.
(\href{http://dx.doi.org/10.1103/PhysRevLett.119.050504}{10.1103/PhysRevLett.119.050504})

\bibitem{Lovasz.1986.Matchinga}
Lov\'asz L, Plummer MD. 1986  Matching and {{Linear Programming}}. In Lov\'asz L, Plummer MD, editors, {\em Matching {{Theory}}}, North-{{Holland Mathematics Studies}},  vol. 121,  pp. 255--305. {North-Holland}.
(\href{http://dx.doi.org/10.1016/S0304-0208(08)73643-0}{10.1016/S0304-0208(08)73643-0})

\bibitem{Weyl.1952.elementary}
Weyl H. 1952  The Elementary Theory of Convex Polyhedra. In Kuhn HW, Tucker AW, editors, {\em Contributions to the {{Theory}} of {{Games}}} ,  number~24 in Annals of {{Mathematics Studies}} pp. 3--18. {Princeton University Press}.
(Translation of Weyl H. 1934. Elementare Theorie der konvexen Polyeder, \emph{Comment. Math. Helvetici} \textbf{7}, 290--306).

\bibitem{Cervantes.2020.Contextuality}
Cervantes VH, Dzhafarov EN. 2020  Contextuality {{Analysis}} of {{Impossible Figures}}. {\em Entropy} \textbf{22}, 981.
(\href{http://dx.doi.org/10.3390/e22090981}{10.3390/e22090981})

\bibitem{cervantes2023hypercyclic}
Cervantes VH, Dzhafarov EN. 2023  Hypercyclic systems of measurements. {\em arXiv preprint arXiv:2304.01155}.

\bibitem{Cervantes.2021.Contextuality}
Cervantes VH, Dzhafarov EN. 2020  Contextuality analysis of impossible figures. {\em Entropy} \textbf{22}, 981.
(\href{http://dx.doi.org/10.3390/e22090981}{10.3390/e22090981})

\end{thebibliography}
\end{document}